\documentclass[twocolumn,10pt]{IEEEtran}

\usepackage{braket}

\renewcommand{\baselinestretch}{0.79}

\input{def.tex}
\usepackage{dsfont}
\usepackage{minibox}
\DeclareSymbolFont{matha}{OML}{Txmi}{m}{it}
\DeclareMathSymbol{\varv}{\mathord}{matha}{118}
\usepackage{eurosym}
\usepackage{hhline,physics}
\usepackage{algorithm,algorithmicx}
\usepackage{multicol}\IEEEoverridecommandlockouts
\begin{document}
	\title{
		Coverage Probability of STAR-RIS assisted Massive MIMO systems with   Correlation and Phase Errors}
	\author{Anastasios Papazafeiropoulos, Zaid Abdullah, Pandelis Kourtessis, Steven Kisseleff, Ioannis Krikidis
		\thanks{This work was co-funded by the European Regional Development Fund
			and the Republic of Cyprus through the Research and Innovation Foundation,
			under the projects INFRASTRUCTURES/1216/0017 (IRIDA) and EXCEL-
			LENCE/0918/0377 (PRIME). Also, this work was supported by the Luxembourg National Research Fund (FNR) through the CORE Project under Grant RISOTTI C20/IS/14773976.}
		\thanks{A. Papazafeiropoulos is with the Communications and Intelligent Systems Research Group, University of Hertfordshire, Hatfield AL10 9AB, U. K., and with SnT at the University of Luxembourg, Luxembourg.  	P. Kourtessis is with the Communications and Intelligent Systems Research Group, University of Hertfordshire, Hatfield AL10 9AB, U. K. Zaid Abdullah and Steven Kisseleff are  with the SnT at the University of Luxembourg, Luxembourg.  I. Krikidis is with the IRIDA Research Centre for Communication Technologies, Department of Electrical and Computer Engineering, University of Cyprus, Cyprus. E-mails: tapapazaf@gmail.com, \{zaid.abdullah, steven.kisseleff\}@uni.lu, p.kourtessis@herts.ac.uk,krikidis.ioannis@ucy.ac.cy. }}
	\maketitle\vspace{-1.7cm}
	\begin{abstract}
		In this paper, we investigate a  simultaneous transmitting and reflecting reconfigurable intelligent surface (STAR-RIS) assisting a massive multiple-input multiple-output (mMIMO) system. In particular, we derive a   closed-form expression for the coverage probability of a STAR-RIS assisted mMIMO system while accounting for correlated fading and phase-shift errors. Notably, the phase configuration takes place at every several coherence intervals  by optimizing the  coverage probability since the latter depends  on statistical channel state information (CSI) in terms of large-scale statistics. As a result, we achieve a reduced complexity and overhead for the optimization of passive beamforming, which are increased in the case of STAR-RIS networks with instantaneous CSI. Numerical results corroborate our analysis, shed light on interesting properties such as the impact of the number of RIS elements and the effect of phase errors, along with affirming the superiority of STAR-RIS against reflective-only RIS.
	\end{abstract}
	\section{Introduction}
	Reconfigurable intelligent surface (RIS) is a promising technology, enabling a smart radio environment (SRE) to meet the challenges of sixth-generation (6G) wireless networks \cite{DiRenzo2020}. It consists of a number of low-cost  nearly passive elements that can modify the phase shifts and even the amplitude of the impinging signals to realize an SRE through a backhaul controller.  RIS favorable characteristics have led recently the research interest in meeting certain significant targets such as maximization of the spectral and energy efficiencies \cite{Papazafeiropoulos2021,Huang2019},   and maximization of the coverage probability \cite{Guo2020,Papazafeiropoulos2021a}. Notably, given that 6G will be built upon existing technologies in 5G such as massive multiple-input multiple-output (mMIMO) systems, the latter suggests a promising architecture for amalgamation with RIS \cite{Zhi2021}.
	
	Thus far, there has been a considerable amount of research on RISs. However, most of the existing works in the literature assume independent Rayleigh fading such as \cite{Wu2019a,Huang2019}. In contrast, it was shown in  \cite{Bjoernson2020}, it was shown that RIS correlation should be taken into account in practice. On this ground, several recent research efforts have considered RIS correlation in their analysis \cite{Papazafeiropoulos2021a,Papazafeiropoulos2022}. Furthermore, aiming at realistic modeling, many works have taken into account phase-shift errors coming from the finite precision of  phase-shifts configuration   \cite{Badiu2019,Papazafeiropoulos2021}.

	Recent advancements on metasurfaces have brought to the forefront the concept of simultaneous transmitting and reflecting RISs (STAR-RISs). In particular, the STAR-RIS can provide a full space coverage by not only reflecting in the half space, but also refracting to space behind the RIS \cite{Xu2021,Mu2021a}.  In particular, in \cite{Xu2021}, a general 	hardware model and two-channel models for the near-field
	region and the far-field region of STAR-RIS have been presented while showing that their diversity gain and coverage are greater than conventional RIS (i.e. reflective-only RIS) assisted systems. Moreover, in \cite{Mu2021a}, three operating protocols for STAR-RIS were suggested, which are known as energy splitting (ES), mode switching (MS), and time switching (TS). 
	However, none of the existing works on STAR-RIS have  considered correlated fading, phase errors, and mMIMO systems.
	
	Against the above background, we provide in closed form the only work obtaining the coverage probability for STAR-RIS assisted mMIMO systems with practical channel and phase models.  In particular, contrary to existing works on STAR-RIS such as\cite{Xu2021,Mu2021a}, we formulate a model that embodies correlated fading and phase errors to identify the realistic prospects of STAR-RIS before its final implementation. In this respect, we derive the coverage probability for mMIMO systems for both ES and MS protocols. Especially, compared to \cite{Xu2021}, which assumed a single-antenna transmitter, and required no special phase-shifts optimization,   we also consider a large number of antennas at the  base station (BS). For this reason,   we follow the methodology in \cite{Papazafeiropoulos2021,Papazafeiropoulos2021a} based on statistical  channel state information (CSI) to optimize the passive (reflecting and refracting) beamforming matrix (PBM) of each user equipment (UE) at every several coherence intervals, which brings lower overhead compared to optimizations relying on instantaneous CSI. To the best of our knowledge, this is the only work on STAR-RIS, where both beamforming matrices for reflection and refraction are optimized based on  statistical CSI, and thereby reducing the amount of required overhead, which can be considered as one of the main challenges  for STAR-RIS.
	
	\textit{Notation}: Vectors and matrices are described by boldface lower and upper case symbols, respectively. The notations $(\cdot)^\T$, $(\cdot)^\H$, and $\tr\!\left( {\cdot} \right)$ denote the transpose, Hermitian transpose, and trace operators, respectively. Moreover, the notations $ \arg\left(\cdot\right) $,  $\EE\left[\cdot\right]$,  and $ \mathrm{Var}(\cdot) $ denote the argument function,  the expectation, and variance operators, respectively. The notation  $\diag\left(\bA\right) $ denotes a vector with elements equal to the  diagonal elements of $ \bA $, while  $\bb \sim \cC\cN{(\b0,\mathbf{\Sigma})}$ denotes a circularly symmetric complex Gaussian vector with zero mean and covariance matrix $\mathbf{\Sigma}$. 
	\section{Signal and System Models}	
	This section presents the signal and system models of the STAR-RIS assisted system.
	\subsection{Signal Model}
	Regarding the description of the signal model of the STAR-RIS, let $ s_{n} $ describe the incident signal on element $ n\in  \mathcal{N} $, where $ \mathcal{N}=\{1, \ldots, N\} $  is the set of RIS  elements. Two independent coefficients, denoted as the transmission and the 	reflection coefficients, configure the transmitted 	and reflected signals in respective modes. Especially, the transmitted ($ t $) and reflected ($ r $) signals by the $ n $th element can be modelled as $ t_{n} =(\sqrt{\beta_{n}^{t}}e^{j \phi_{n}^{t}})s_{n}$ and $ r_{n}=(\sqrt{\beta_{n}^{r}}e^{j \phi_{n}^{r}})s_{n} $, respectively, where $ {\beta_{n}^{k}}\in [0,1] $ and $ \phi_{n}^{k} \in [0,2\pi)$  express the independent amplitude and phase-shift response of the $ n $th element, and $  k\in\{t,r\} $ corresponds to the UE found in the transmission ($ t $) or reflection  ($ r $) region 
	\cite{Xu2021}.\footnote{In practical applications, the amplitude and phase-shifts are correlated and result in a performance loss. The study of this correlation is  left for future work.} 
	Note that the  choice of  $ \phi_{n}^{t} $ and $ \phi_{n}^{r} $ is independent from each other, but the amplitude adjustments are correlated based on the law of energy conservation as
	\begin{align}
		\beta_{n}^{t}+\beta_{n}^{r}=1,  \forall n \in \mathcal{N}.
	\end{align}
	\textbf{\textit{Operation Protocols}}:
	Herein, we present briefly the main points of the ES/MS protocols  \cite{Mu2021a}. The study of the TS protocol is left for future work.
	
	\textit{{ES protocol:}} All RIS elements serve simultaneously  $ t $ and $ r $ UEs,  and the corresponding PBM for $ k\in\{t,r\}
	$ is expressed as $ \bPhi_{k}^{\mathrm{ES}}=\diag(\sqrt{\beta_{1}^{k}}e^{j \phi_{1}^{k}}, \ldots, \sqrt{\beta_{N}^{k}}e^{j \phi_{N}^{k}}) \in \mathbb{C}^{N\times N}$, where $ \beta_{n}^{t} , \beta_{n}^{r} \in [0,1]$, $ 	\beta_{n}^{t}+\beta_{n}^{r}=1 $, and $ \phi_{n}^{t}, \phi_{n}^{r}\in [0, 2\pi ), \forall n \in \mathcal{N} $.
	
	\textit{{MS protocol:}} The elements are divided into two groups of $ N_{t} $ and $ N_{r} $ elements serving UE  $  t$ or $ r $, respectively, i.e., $ N_{t}+N_{r}=N $. In such case, the PBM for $  k\in\{t,r\} $
	is  $ \bPhi_{k}^{\mathrm{MS}}=\diag(\sqrt{\beta_{1}^{k}}e^{j \phi_{1}^{k}}, \ldots, \sqrt{\beta_{N}^{t}}e^{j \phi_{N}^{k}}) \in \mathbb{C}^{N\times N}$, where $ \beta_{n}^{t} , \beta_{n}^{r} \in \{0,1\}$, $ 	\beta_{n}^{t}+\beta_{n}^{r}=1 $, and $ \phi_{n}^{t}, \phi_{n}^{r}\in [0, 2\pi ), \forall n \in \mathcal{N} $. Basically, this protocol can be considered as a special case of ES by restricting the amplitude coefficients for transmission and reflection to binary values. Thus, the MS protocol cannot provide the full-dimension transmission and reflection beamforming gain as ES, however, it requires less computational complexity in terms of PBM design.
	
	\subsection{System Model}
	We consider a STAR-RIS-assisted mMIMO communication system, where a BS with a large number of antennas, denoted by $ M $, communicates with $ 2 $ single-antenna UEs, as shown in Fig \ref{Fig0}.\footnote{This model can be easily extended to the multi-user scenario as in typical mMIMO, where all UEs are divided into multiple groups of  two UEs located at the opposite sides of RIS.}
Especially,  the incident signals on the RIS are divided simultaneously into transmitted and reflected signals to create  full-space coverage that are received by UE $ t $  located behind the RIS, and UE $r $  located in front of the RIS, i.e., the BS and UE $ r $ are found at the same side. Notably, we assume that the UEs operate at orthogonal frequency bands of equal sizes similar to \cite{Xu2021}. 
Also, we account for no direct links between the BS and the users due to blockages, which is one of the main challenging practical scenarios that justifies the use of a RIS. The phase shifts of both transmitted and reflected signals are adjusted by a controller that exchanges information with the  BS through a backhaul link.  We rely on perfect CSI, which means that the results play the role of upper bounds of practical scenarios with imperfect CSI. This is a common assumption on works studying coverage to allow more direct mathematical manipulations and to focus on the transmission and reflection beamforming gain of STAR-RIS, 
\begin{figure}[!h]
	\begin{center}
		\includegraphics[width=0.65\linewidth]{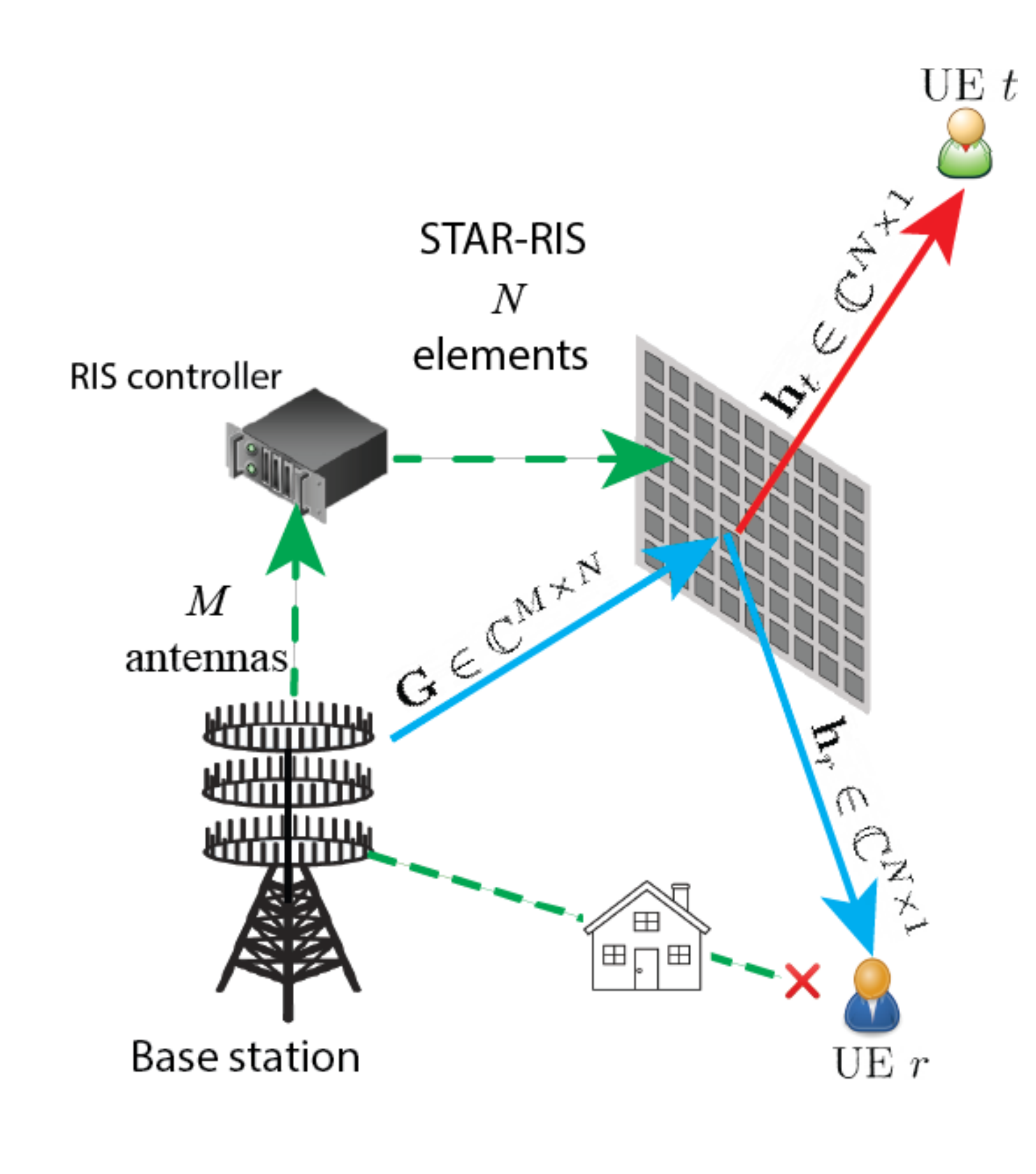}
		\caption{\footnotesize{A STAR-RIS assisted mMIMO model. }}
		\label{Fig0}
	\end{center}
\end{figure} 	

We assume narrow-band quasi-static fading channels. Specifically, $ \bh_{k} \in \mathbb{C}^{N \times 1}$
describes the channel between the RIS and UE $ k\in \{t,r\} $ while  $\bG=[\bg_{1}\ldots,\bg_{N} ] \in \mathbb{C}^{M \times N}$, expresses the LoS channel between the BS and RIS  with $ \bg_{i}, $  $ i=1 \in \mathcal{N} $ corresponds to the $ i $th column vector of $ \bG$. In particular, the $ (m,n) $th entry of $ \bG$ can be expressed as
\begin{align}
	[\bG]_{m,n}&=\sqrt{\beta_{g}} \exp \Big(j \frac{2 \pi }{\mu}\left(m-1\right)d_{\mathrm{BS}}\sin \theta_{1,n}\sin \psi_{1,n}\nn\\
	&+\left(n-1\right)d_{\mathrm{RIS}}\sin \theta_{2,m}\sin \psi_{2,m}\Big).
\end{align}
Note that $ \bG$ could be obtained by several techniques, e.g., see \cite{Bohagen2007}. The parameters  $ \beta_{g} $ and $ \lambda $ are the path-loss of the BS-RIS link and the carrier wavelength. Also, $ d_{\mathrm{BS}} $ and $ d_{\mathrm{RIS}} $ are the inter-antenna separation at the BS and inter-element separation at the RIS, respectively \cite{Papazafeiropoulos2021}. In addition, $ \theta_{1,n} $, $ \psi_{1,n} $  and  $ \theta_{2,m} $, $ \psi_{2,m} $ describe the  elevation and azimuth LoS angles of departure  at the BS and  arrival  at the RIS  with respect to RIS element $ n $, respectively. 

\subsubsection{RIS correlation}
Instead of assuming independent Rayleigh fading as in existing STAR-RIS works \cite{Xu2021,Mu2021a}, we consider spatial correlation, appearing in practice \cite{Bjoernson2020}.\footnote{The consideration of Rician fading, which includes an LoS component, is the topic of future work.} Hence, the  channel vector between the RIS and UE $ k $  with $ k\in \{t,r\} $ can be written  as
\begin{align}
	\bh_{k}&=\sqrt{\beta_{k}}\bR_{\mathrm{RIS}}^{1/2}\bz_{k},
\end{align}
where  $ \beta_{k} $ 
expresses the path-loss,  and $ \bR_{\mathrm{RIS}} \in \mathbb{C}^{N \times N} $ with $ \tr\left(\bR_{\mathrm{RIS}} \right)=N $ expresses the deterministic Hermitian-symmetric positive semi-definite correlation
matrix at the RIS. In particular, if  $N=N_{\mathrm{H}}N_{\mathrm{V}}  $ elements with $ N_{\mathrm{H}} $ the number of elements per row and $ N_{\mathrm{V}} $ the number of elements per column, while $ d_\mathrm{H} $ and $ d_\mathrm{V} $ are horizontal and vertical dimensions of each elements, the $ (i,j) $th element of the RIS correlation is given by \begin{align} \label{eq:Element}
r_{ij} = d_{\mathrm{H}} d_{\mathrm{V}} \mathrm{sinc} \left( 2 \|\mathbf{u}_{i} - \mathbf{u}_{j} \|/\lambda\right),
\end{align}
where $\mathbf{u}_{\epsilon} = [0, \mod(\epsilon-1, N_\mathrm{H})d_\mathrm{H}, \lfloor (\epsilon-1)/N_\mathrm{H} \rfloor d_\mathrm{V}]^\T$, $\epsilon \in \{i,j\}$.

Note that the path-losses and the correlation matrices are assumed known after being obtained by practical methods, {e.g., see} \cite{Neumann2018}. Moreover,  $ \bz_{k}\sim \mathcal{CN}\left(\b0,\Id_{N}\right) $ 
describes the independent and identically distributed (i.i.d.)  fast-fading vector in the $ k\mathrm{th} $ link.
\subsubsection{Phase-shift errors}
The configuration of the RIS elements with infinite precision is not possible in practice. Hence, phase errors appear, which  result from imperfections in phase-estimation and/or phase quantization that cannot be avoided \cite{Badiu2019}. Specifically, we denote $ \tilde{\phi}_{n}^{k}$, with $n \in \mathcal{N} $, $ k\in\{t,r\} $ the  i.i.d. randomly distributed  phase error in  $ [-\pi, \pi) $ of the $ n $th phase-shift aiming at the $ k $th UE. Hence, the  random phase error matrix corresponding to the PBM $  \bPhi_{k} $  is diagonal and described by $\widetilde{\bPhi}_{k} =\diag\left( e^{j \tilde{\phi}_{1}}, \ldots, e^{j \tilde{\phi}_{N}} \right)\in\mathbb{C}^{N\times N}$. 	The probability density function (PDF) of $\tilde{\phi}_{n}^{k} $ is assumed symmetric with its mean direction equal to zero, i.e., $ \arg\left(\EE[\mathrm{e}^{j\tilde{\phi}_{n}^{k}}]\right)=0 $ \cite{Badiu2019}.

In general, the Uniform and the Von Mises distributions are the main PDFs that can describe the phase noise affecting the  RIS \cite{Badiu2019}. Notably, the Uniform PDF of phase-noise has a  characteristic function (CF) equal to zero, and thus it cannot provide any knowledge regarding the phase-estimation accuracy. 
In contrast,   the Von Mises PDF  has a zero-mean and concentration parameter $ \kappa_{\tilde{\phi}} $, which captures the accuracy of the estimation. In particular, the CF of the Von Mises noise  CF is $ m=\frac{\mathrm{I}_{1}\!\left(\kappa_{\tilde{\phi}}\right)}{\mathrm{I}_{0}\!\left(\kappa_{\tilde{\phi}}\right)}$, where $ \mathrm{I}_{v}\!\left(\kappa_{\tilde{\phi}}\right)$ is the modified Bessel function of the first kind and order 	$ v $.

The  received signal by UE $ k\in \{t,r\} $   through the STAR-RIS assisted network is described by
\begin{align}
	y_{k}=	\sqrt{\rho_{\mathrm{dl}}}\bh_{k}^{\H}\widetilde{\bPhi}_{k}^{\H,\X}	\bPhi_{k}^{\H,\X}\bG^{\H} \bff_{k}q_{k}+z_{k},\label{received}
\end{align}
where $ \rho_{\mathrm{dl}}\ge 0 $ is the  transmit power at BS to  UE $ k $, and $ q_{k} $ is the data symbol with $ \EE[|q_{k}|^{2}]=1 $, while $ \bff_{k}\in \mathcal{C}^{M \times 1} $ is the linear precoding vector, and $ z_{k}\sim \mathcal{CN}(0, N_{0}) $ is the additive white Gaussian noise with zero mean and variance of  $ N_{0} $.

Notably, the cascaded channel vector including the phase-shift errors between the large BS and UE $ k\in\{t,r\} $ can be expressed as $ \bar{\bh}_{k}=\bG \bPhi_{k}^{\X} \widetilde{\bPhi}_{k}^{\X}	\bh_{k} \in \mathbb{C}^{M\times 1} $, where $ \mathcal{X}=\mathrm{ES},\mathrm{MS} $. It is distributed as $ \mathcal{CN}(\b0, \bar{\bR}_{k}) $, where $ \bar{\bR}_{k}= \beta_{k}\bG  \bPhi_{k}^{\X}\tilde{\bR}_{\mathrm{RIS}}^{\X} \bPhi_{k}^{\H,\X}\bG^{\H}$. Note that $ \tilde{\bR}_{\mathrm{RIS}}^{\X}=\EE[ \bPhi_{k}^{\X}\bR_{\mathrm{RIS}}\bPhi_{k}^{\H,\X}]= m^{2}\bR_{\mathrm{RIS}}+\left(1-m^{2}\right)\Id_{N}$, where $ m $ denotes the CF of the phase error \cite[Eq. 13]{Papazafeiropoulos2021}. In the case of the Uniform distribution, i.e, $  m=0 $,  $ \bR_{k} $ does not depend on the phase-shifts and  RIS cannot be optimized. A similar observation is met when RIS optimization is based on statistical CSI \cite{Papazafeiropoulos2021}. Hence, apart from its practical meaning, RIS correlation should be taken into account to exploit RIS.

\section{Main Results}\label{MainResults}
This section presents  the signal-to-noise ratio (SNR), the coverage probability, and its optimization for STAR-RIS assisted systems. Generally, the coverage probability for $ k$th  UE, $ {P}_{\mathrm{c}}^{k} $  with $ k\in \{t,r\} $, is  the probability that the received SNR is greater than a threshold  $ T $, i.e., it is described by $ {P}_{\mathrm{c}}^{k}=\mathrm{Pr}\left(\gamma_{k}>T\right) $, where $ \mathrm{Pr}(\cdot) $ denotes probability.


\subsection{SNR}	
First, we take advantage of channel hardening,  because UEs do not have any knowledge of the instantaneous CSI in practice, but they are aware of their statistics \cite{Bjoernson2017}. Note that  channel hardening appears in mMIMO systems as the number of BS antennas increases. Hence, by resorting  to the application of the use-and-then-forget bounding technique \cite{Bjoernson2017}, we write \eqref{received} as
\begin{align}
	y_{k}&=	\sqrt{\rho_{\mathrm{dl}}}\big(\EE[\bar{\bh}_{k}^{\H}\bff_{k}]q_{k}+\bar{\bh}_{k}^{\H}\bff_{k}q_{k}-	\EE[\bar{\bh}_{k}^{\H}\bff_{k}]q_{k}\big)+z_{k}.\label{received1}
\end{align}

Next, by assuming that \eqref{received1} represents a single-input single-output (SISO) system, where the BS treats the unknown terms as uncorrelated additive noise, the achievable SNR of the link between the BS and UE  $ k\in \{t,r\} $ is written as 
\begin{align}
	\gamma_{k} =\frac{\big|\EE[\bar{\bh}_{k}^{\H}\bff_{k}]
		\big|^{2}}{\EE[\big|\bar{\bh}_{k}^{\H}\bff_{k}
		\big|^{2}]-\big|\EE[\bar{\bh}_{k}^{\H}\bff_{k}]
		\big|^{2}+{\sigma_{0}}{}},\label{SNR}
\end{align}
where $ \sigma_{0}=N_{0}/\rho_{\mathrm{dl}} $ \cite{Bjoernson2017}.
Regarding the selection of precoding, we apply the simple  maximum ratio transmission (MRT) precoding because it allows the derivation of closed-from expressions and extraction of fundamental properties together with its optimality for the single UE case, i.e., $ \bff_{k}= \frac{\bar{\bh}_{k} }{ \sqrt{\EE[\|\bar{\bh}_{k}\|^{2}]}}$. \footnote{Another common choice could be the more robust regularized zero-forcing (RZF) precoder, but that could be applicable in a multi-user scenario with interference, and it  would result in  intractable expressions. Hence, the application of RZF could be the topic of future work.}
\begin{proposition}\label{proposition:SNR}
	The downlink SNR of a STAR-RIS assisted mMIMO system with correlated Rayleigh fading and phase-shift errors at UE $ k\in\{t,r\} $ 
	is given by
	\begin{align}
		&\!\!\!	\gamma_{k}\!=\! \frac{\tr^{2}( \bar{\bR}_{k})}{\tr( \bar{\bR}_{k}^{2})+{\sigma_{0}}\tr( \bar{\bR}_{k})}
		\label{DE_SNR} .
	\end{align}
\end{proposition}

\begin{proof}
	See Appendix~\ref{DESNR}.
\end{proof}

\begin{remark}
	As can be seen, the  SNR in \eqref{DE_SNR} depends only on the PBM and the statistical CSI by means of the correlation matrix, the path-losses, and the phase-shift errors. Also, there is a dependence on the location of UE $ k $, i.e., if it is found behind or in front of the RIS. 
\end{remark}

\subsection{Coverage Probability}	\label{coverage1}
Having obtained the SNR, the following proposition provides the coverage probability.

\begin{proposition}\label{coverage}
	The coverage probability of a mMIMO STAR-RIS assisted system, accounting for  RIS correlation and  phase-shift errors of UE $ k\in\{t,r\} $, is tightly approximated as
	\begin{align}
		{P}_{\mathrm{c}}^{k}&\approx \sum^{L}_{n=1}\binom{L}{n}(-1)^{n+1}e^{-\eta \frac{T}{\gamma_{k}}}\label{PC},
	\end{align} 
	where  $\eta=L\left(L! \right)^{-\frac{1}{L}}$ with $ L $ being an approximation parameter.
\end{proposition}
\begin{proof}
	See Appendix~\ref{coverageProof}.
\end{proof}

The coverage depends  on the threshold, the number $ L $ defining the tightness of the approximation, and of course, the downlink SNR with its involved parameters. 

\textit{Conventional RIS:} We  assume a smart surface that consists of
transmitting-only or reflective-only elements, each with $ N/2 $ elements with $ N $ even for  simplicity. Also, $ \beta_{n}^{t} =1$ or $ \beta_{n}^{r} =1, \forall n$, respectively. In Sec. \ref{Numerical}, we compare the performance of
the STAR-RIS with the conventional RIS and illustrate the advantage of adopting  STAR-RIS in wireless systems.

\subsection{Passive beamforming matrix optimization}\label{RISbeamformingmatrixoptimization}
A STAR-RIS assisted system serves  two regions simultaneously, where UEs  $ t $ and $ r $ are found. We focus on the optimization of the total coverage probability given by $ {P}_{\mathrm{c}}^{t} +{P}_{\mathrm{c}}^{r} $. Given the difficulty in simultaneous optimization of both $ {P}_{\mathrm{c}}^{t} $ and $ {P}_{\mathrm{c}}^{r} $, we rely on alternating optimization by optimizing first each of them with respect to its PBM while fixing the other in an iterative manner until reaching the convergence, i.e., we perform optimization of $ {P}_{\mathrm{c}}^{k} $ with respect to $ \bPhi_{k}^{\X} $. Each optimization is achieved in terms of the  projected gradient 	ascent until converging to a 	stationary point.

Hence, based on the common assumption of infinite-resolution phase-shifters, the maximization algorithm for $ {P}_{\mathrm{c}}^{k} $, $ k\in\{t,r\} $ UE with respect to the PBM is formulated as
\begin{align}\begin{split}
		&\!\!\!(\mathcal{P}1)~\max_{\bPhi_{k}^{\X}} ~~{P}_{\mathrm{c}}^{k}\\
		&~~~~~\mathrm{s.t}~~~~\beta_{n}^{t}+\beta_{n}^{r}=1,  \forall n \in \mathcal{N}\\
		&~~~~~~~~~~~~\beta_{n}^{t}, \beta_{n}^{r}\in [0, 1],\phi_{n}^{t},\phi_{n}^{r}\in [0,2\pi),~~  \forall n \in \mathcal{N}.
	\end{split}\label{Maximization} 
\end{align}
By replacing $  [0, 1] $ with $ \{ 0, 1\} $, $ (\mathcal{P}1) $ can be applied  to the MS scheme.

The optimization problem $ 	(\mathcal{P}1) $ is non-convex regarding $ \phi_{n}^{k} $. The dependence on $ \phi_{n}^{k} $ is found on the covariance matrices $ \bar{\bR}_{k}   $.  Given that it is a constrained maximization problem, we resort to  the  projected gradient 	ascent (PGA) until convergence to a 	stationary point to provide its solution \cite{Papazafeiropoulos2021}. The transmit power constraint guarantees its convergence. It is worthwhile to mention that the complexity of $ 	{P}_{\mathrm{c}}^{k} $ is  $ \mathcal{O}\left(G(MN^{2}+M^3+L)\right) $. As can be seen, it is  a function of the fundamental system parameters  $ M $, $ N $, and $L $ with the number of BS antennas having the higher impact. Since it has  to be performed twice for each  $ k\in\{t,r\} $, the complexity is double, while, in the case of the conventional RIS scenario, the complexity is  half of the STAR-RIS setting.

Based on PGA, let $ \bs_{k,l} =[\phi_{1}^{k,l}, \ldots, \phi_{N}^{k,l}]^{\T}$ denote 	the  vector including the phases at step $ l $, while the next iteration point  provides the increase of $ {P}_{\mathrm{c}}^{k} $ upon its convergence by 
projecting the solution onto the closest feasible point as specified by $ \min_{\phi_{n}^{k} \in [0,2\pi), }\|\bs_{k}-\tilde{\bs}_{k}\|^{2} $. The next iteration  is described by
\begin{align}
	\tilde{\bs}_{k,l+1}&=\bs_{k,l}+\mu \bv_{k, l},\label{sol1}\\
	\bs_{k,l+1}^{w}&=\beta_{n}^{t}\exp\left(j \arg \left(\tilde{\bs}_{k,l+1}\right)\right),\label{sol2}
\end{align}
where the parameter $ \mu $ is the step size obtained at each iteration through the backtracking line search \cite{Boyd2004}. Moreover,  $ \bv_{k,l} $ is the  ascent direction at step $ l $, which means  $ \bv_{k,l}= \pdv{	{P}_{\mathrm{c}}^{k}}{\bs_{k,l}^{*}} $. The following lemma provides this derivative.

\begin{lemma}\label{deriv1}
	The derivative of the coverage probability $ {P}_{\mathrm{c}}^{k} $ for UE $ k\in\{t,r\} $  with respect to $ \bs_{k,l}^{*}$ is given by \eqref{lemm1} at the top of the next page, where 
	$ \gamma_{k} $ is given by \eqref{DE_SNR}, and
	\begin{align}
		S_{k}&=\tr^{2}( \bar{\bR}_{k}), \\ 
		I_{k}&=\tr( \bar{\bR}_{k}^{2})+{\sigma_{0}}\tr( \bar{\bR}_{k}), \\
		\pdv{S_{k}}{\bs_{k,l}^{*}}&=2\beta_{k}\tr( \bar{\bR}_{k})\diag(\bG^{\H}\bG  \bPhi_{k}^{\X}\tilde{\bR}_{\mathrm{RIS}}^{\X} ),\\ \pdv{I_{k}}{\bs_{k,l}^{*}}&=\beta_{k}\tr(\bG^{\H}\bar{\bR}_{k}\bG  \bPhi_{k}^{\X}\tilde{\bR}_{\mathrm{RIS}}^{\X})\nn\\
		&	+{\sigma_{0}}\beta_{k}\tr( \bG^{\H}\bG  \bPhi_{k}^{\X}\tilde{\bR}_{\mathrm{RIS}}^{\X}).\label{int2}
	\end{align}
	\begin{figure*}
		\begin{align}
			\!\pdv{{P}_{\mathrm{c}}^{k} }{\bs_{k,l}^{*}}\!\!&=\!\frac{\pdv{S_{k}}{\bs_{k,l}^{*}}I_{k}-S_{k}\pdv{I_{k}}{\bs_{k,l}^{*}}}{\gamma_{k}^{2}I_{k}^{2}}\sum^{L}_{n=1} \!\binom{L}{n}\! \left( -1 \right)^{n+1} n \eta T \mathrm{e}^{ -n \eta \frac{T}{\gamma_{k}}}.\label{lemm1}
		\end{align}
		\line(1,0){515}
	\end{figure*}
\end{lemma}
\begin{proof}
	See Appendix~\ref{Derivative}.
\end{proof}

\section{Numerical Results}\label{Numerical} 
In this section, we provide the numerical results for the coverage of the STAR-RIS assisted mMIMO system. We consider a uniform planar array (UPA) of $ N=60 $ elements for the RIS, while a uniform linear array (ULA) with  $ M=40 $ antennas is assumed for the BS that serves UEs $ t $ and $ r $. Unless otherwise stated, we consider the following values. The path-losses are generated according to  the NLOS  version of the 3GPP Urban Micro (UMi) scenario from TR36.814 for a carrier frequency of $ 2.5 $ GHz, and noise level $ -80 $ dBm \cite{3GPP2017}. Specifically, we have $ \beta_{g}=C_{g} d_{g}^{-\nu_{g}},  $ and $ \beta_{k}=C_{k} d_{k}^{-\nu_{k}},~k\in \{t,r\}  $	with $ C_{g}=26 $ dB, $ C_{t}=26 $ dB, $ C_{r}=28 $ dB, $ \nu_{g} =2.1$, $ \nu_{t} =2.5$, and $ \nu_{r} =2.2$. The variables $ d_{g}$ and $ d_{k}$ denote the distances between the BS and RIS and between the RIS and UE $ k\in\{t,r\} $, respectively. We 
assume that the correlation matrix for the  RIS is given by \cite{Bjoernson2020}, where the horizontal and vertical dimension of each element is $d_\mathrm{H}= \lambda/8$ and $d_\mathrm{V}= \lambda/8$. Also,  we assume that $ \beta_{n}^{t}=0.4 $, $ \beta_{n}^{r} =0.6, \forall n$, $ \rho_{\mathrm{dl}}= 6~\mathrm{dB} $, and $ N_{0} =-174+10\log_{10}B_{\mathrm{c}}$ with $ B_{\mathrm{c}}=200~\mathrm{KHz} $. 
Monte-Carlo (MC) simulations are carried out to verify our analysis.

In Fig. \ref{Fig1}, we depict the coverage probability versus the target rate for varying number of RIS elements $ N $ while studying other key properties. We observe that a higher number of $ N $, increases $ {P}_{\mathrm{c}}^{k} $ as expected. Moreover, we study the impact of phase errors for $ N =25$, and we see that, in the case of Uniform PDF, where no knowledge on the errors is known, the RIS cannot be optimized, since the covariance matrix does not depend on the phase shifts. Thus, the coverage is lower. However, if the phase errors follow the Von Mises PDF ($ m =0.5$), $ {P}_{\mathrm{c}}^{k} $ is higher thanks to the possibility of performing phase-optimization under such phase-noise distribution. Also, we make a similar observation regarding the correlation for $ N=100 $, i.e., we notice that a RIS correlation allows the PBM optimization, and $ {P}_{\mathrm{c}}^{k} $ increases  in the case of statistical CSI modeling. In addition, the ES scheme achieves higher coverage than the MS scheme because the latter is a special case of ES. Furthermore, we have added the impact of the phase error for both STAR-RIS and conventional RIS for $ m=0.5 $ and $ m=0.7 $, and we observe that a larger $ m $ increases the estimation accuracy, i.e., the coverage is higher  in both cases, while the coverage for the STR-RIS is generally higher compared to the conventional RIS as expected.

Fig. \ref{Fig2} illustrates the  coverage probability versus the target rate for varying number of BS antennas  $ M$, i.e., $ M=30,60,90 $. It is shown that $ {P}_{\mathrm{c}}^{k} $ increases with an increasing number of BS antennas because of a higher beamforming gain.  Also, we notice that independent Rayleigh fading conditions  result in the worst coverage because no RIS exploitation can be achieved in the case of statistical CSI. In this direction, we observe that correlation should be taken into account to benefit from the RIS, while a higher correlation leads to lower coverage. For the sake of comparison, in the same figure, we have included the coverage in the case of a conventional RIS as presented in Sec. \ref{coverage1} with $ N=30 $ elements. Notably, the STAR-RIS system outperforms the conventional RIS-assisted system because more degrees of freedom for transmission and reflection can be harnessed.

\begin{figure}[!h]
	\begin{center}
		\includegraphics[width=0.85\linewidth]{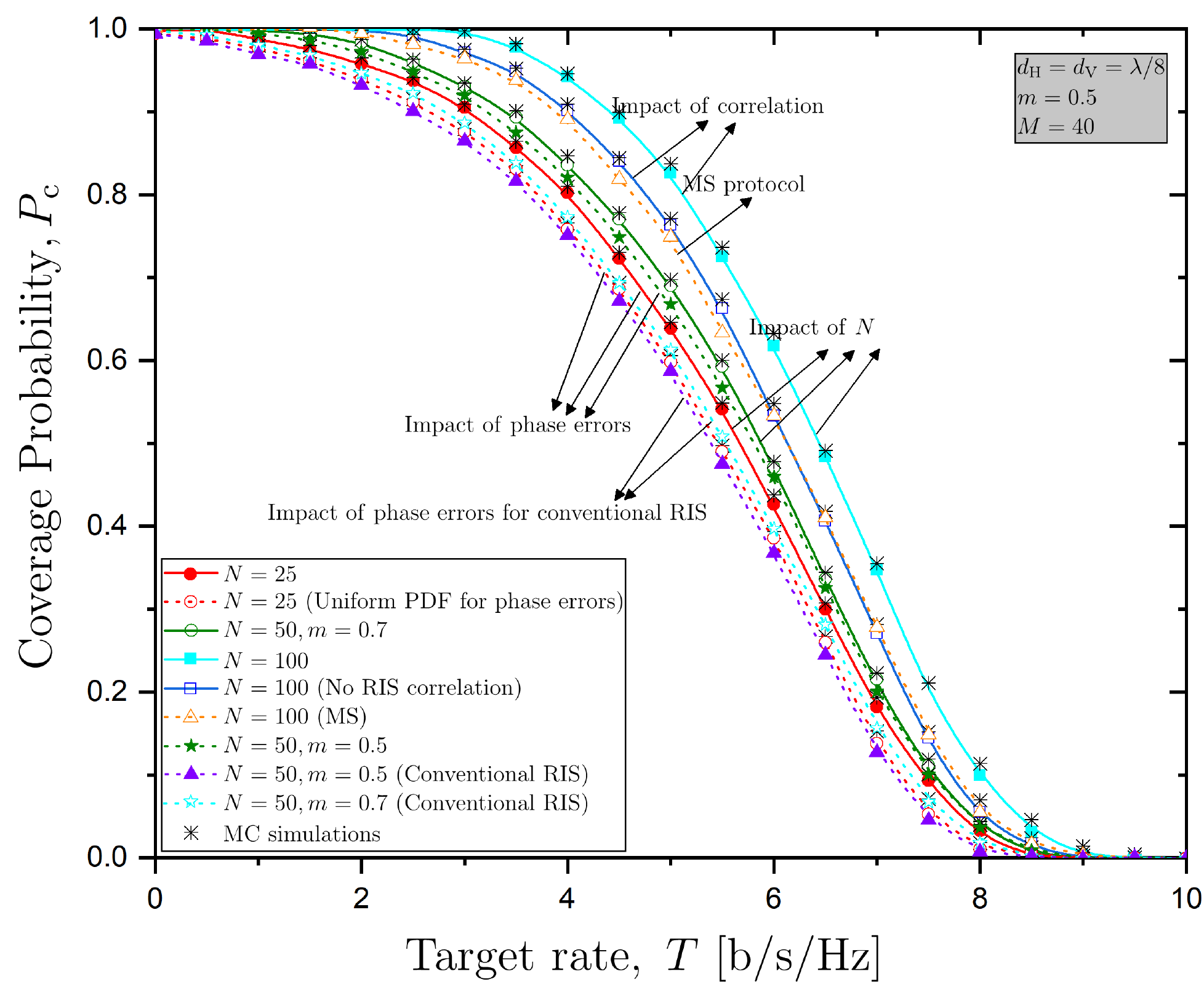}
		\caption{\footnotesize{Coverage probability of a STAR-RIS assisted mMIMO system versus the target rate $ T $ for varying $N$  ($ M=40 $, $ \mathrm{d_{\mathrm{H}}=d_{\mathrm{V}}=\lambda/8} $, $ m=0.5 $). }}
		\label{Fig1}
	\end{center}
\end{figure} 	

\begin{figure}[!h]
	\begin{center}
		\includegraphics[width=0.85\linewidth]{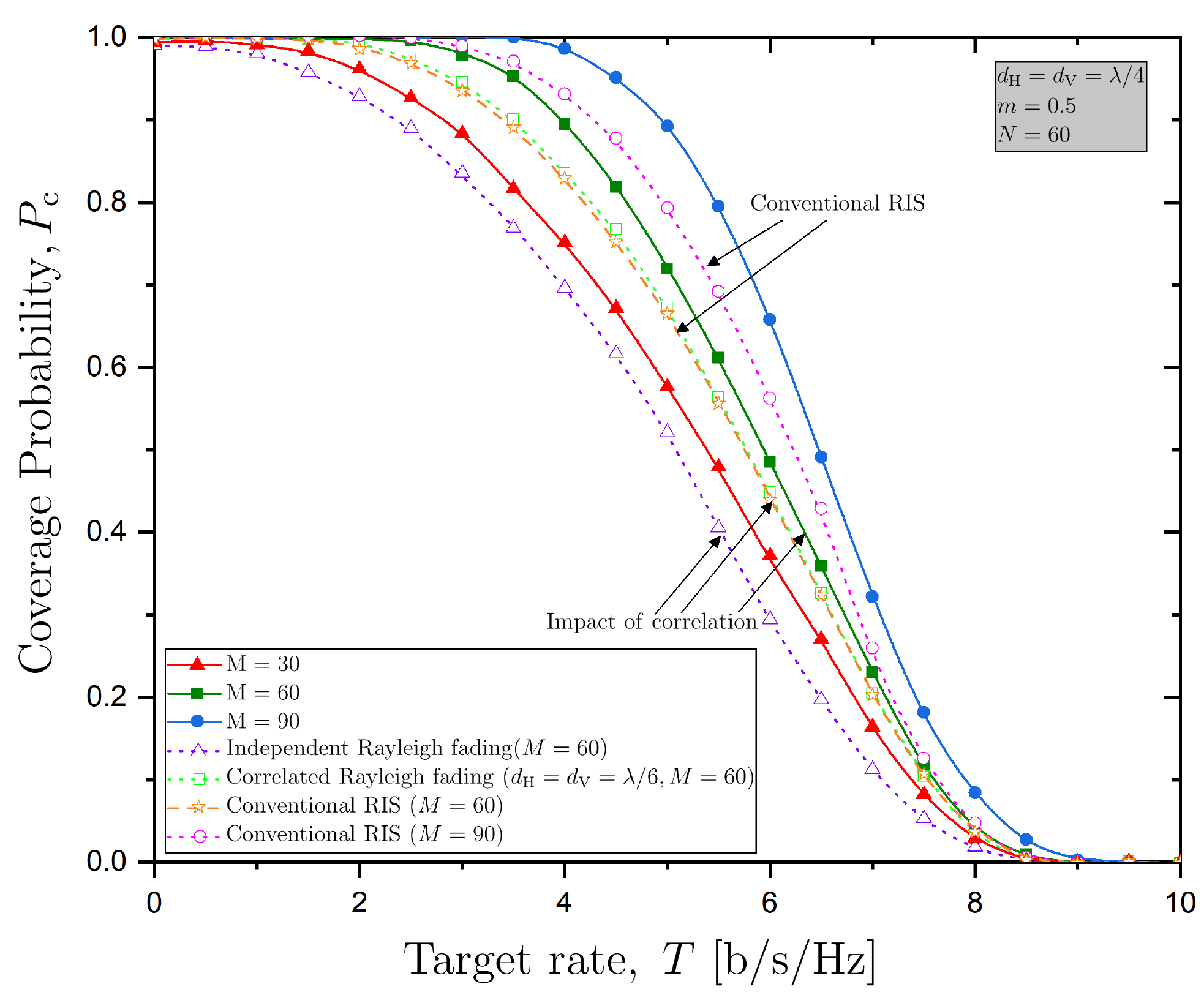}
		\caption{\footnotesize{Coverage probability of a STAR-RIS assisted mMIMO system versus the target rate $ T $ for varying $M$  ($ N=60 $, $ \mathrm{d_{\mathrm{H}}=d_{\mathrm{V}}=\lambda/8} $, $ m=0.5 $). }}
		\label{Fig2}
	\end{center}
\end{figure} 	

\section{Conclusion} \label{Conclusion} 
In this paper, we obtained the coverage probability of a STAR-RIS assisted mMIMO system under the realistic assumptions of correlated Rayleigh fading and phase-shift errors, while previous works on STAR-RIS did not account for these inevitable effects. Especially, we derived the coverage probability for both $ t $ and $ r $ UEs in terms of large-scale statistics that change at every several coherence intervals, and thus, reduce the increased overhead of STAR-RIS. Among others, we depicted how the fundamental parameters affect the coverage, and the outperformance of STAR-RIS compared to reflective-only RIS. Future works could consider Ricean channels or even millimetre-wave transmission. 
\begin{appendices}
	\section{Proof of Proposition~\ref{proposition:SNR}}\label{DESNR}
	Based on the property $ \bx^{\H}\by=\tr(\by \bx^{\H}) $, the term in the numerator of \eqref{DE_SNR} is written as 
	\begin{align}
		\!\!\!|\EE[\bar{\bh}_{k}^{\H}\bff_{k}]|^{2}&
		\!=\!\Big|\!\tr(\!\frac{\EE[\bar{\bh}_{k}\bar{\bh}_{k}^{\H}]}{\EE[\|\bar{\bh}_{k}\|^{2}]}\!)\Big|^{2}\label{term2}\!\\
		&\!=\!\tr( \bar{\bR}_{k}),\label{term4}
	\end{align}
	where, in \eqref{term2}, we have substituted the precoder $ \bff_{k} $.
	Note that $ \bar{\bR}_{k}= \beta_{k}\bG  \bPhi_{k}^{\X}\tilde{\bR}_{\mathrm{RIS}}^{\X} \bPhi_{k}^{\H,\X}\bG^{\H}$.
	
	Regarding the first term in the denominator, we have 
	\begin{align}
		\!\!\!\!\EE[\big|\bar{\bh}_{k}^{\H}\bff_{k}
		\big|^{2}]&\!=\!\tr( \bar{\bR}_{k})+\frac{\tr( \bar{\bR}_{k}^{2})}{\tr( \bar{\bR}_{k})},\label{term1}
	\end{align}
	where \eqref{term1} is obtained by  \cite[Lem.~B.14]{Bjoernson2017}.
		The proof is concluded by appropriate substitutions 
		into \eqref{SNR}.
	\end{appendices}
	
	\section{Proof of Proposition~\ref{coverage}}\label{coverageProof}
	The proof follows similar lines with the technique in \cite{Bai2015}. Specifically, based on the definition of the coverage probability of UE $ k\in\{t,r\} $, we have
	\begin{align}
		{P}_{\mathrm{c}}^{k}&\approx \mathbb{P}(\tilde{g}>\frac{T}{\gamma_{k}})\label{term6}\\
		&\approx 1-(1-e^{-\eta \frac{T}{\gamma_{k}}})^{L}\label{term7}\\
		&=\sum_{n=1}^{L}\binom{L}{n}(-1)^{n+1}e^{-\eta \frac{T}{\gamma_{k}}}\label{term8}.
	\end{align}
	
	In \eqref{term6}, we have used the approximation \cite{Bai2015}, where the constant number $ 1 $  can be replaced by a dummy Gamma variable $ \tilde{g}		 $, which has a unit mean and shape parameter $ L $,  i.e., its pdf is $ \frac{L^{-L}}{\Gamma(L)} \tilde{g}^{L-1}e^{-{L \tilde{g}}{}}$. The tightness of this approximation increases with increasing $L $. In \eqref{term7}, we have applied  Alzer's inequality, where $\eta=L\left(L! \right)^{-\frac{1}{L}}$ \cite{Bai2015}, while in \eqref{term8}, we have applied the Binomial theorem.
	\section{Proof of Lemma~\ref{deriv1}}\label{Derivative}
	First, we apply the  chain rule, which gives
	\begin{align}
		\pdv{{P}_{\mathrm{c}}^{k}}{\bs_{k,l}^{*}}&=\pdv{{P}_{\mathrm{c}}^{k}}{\gamma_{k}}\pdv{\gamma_{k}}{\bs_{k,l}^{*}}.\label{deriv2}
	\end{align}
	Regarding the first term in \eqref{deriv2}, we have
	\begin{align}
		\pdv{{P}_{\mathrm{c}}^{k}}{\gamma_{k}}&=\sum^{L}_{n=1} \!\binom{L}{n}\! \frac{\left( -1 \right)^{n+1} n \eta T}{\gamma_{k}^{2}} \mathrm{e}^{ -n \eta \frac{T}{\gamma_{k}}}.\label{deriv3}
	\end{align}
	For the derivation of the second derivative in \eqref{deriv2}, i.e., $ \pdv{\gamma_{k}}{\bs_{k,l}^{*}}  $,   we define $ \gamma_{k}= \frac{S_{k}}{I_{k}}$, where $ S_{k}=\tr^{2}( \bar{\bR}_{k}) $ and $ I_{k}=\tr( \bar{\bR}_{k}^{2})+{\sigma_{0}}\tr( \bar{\bR}_{k}) $. Hence, the derivative can be written as
	\begin{align}
		\pdv{\gamma_{k}}{\bs_{k,l}^{*}}=\frac{\pdv{S_{k}}{\bs_{k,l}^{*}}I_{k}-S_{k}\pdv{I_{k}}{\bs_{k,l}^{*}}}{I_{k}^{2}}.\label{gam1}
	\end{align}
	The computation of $ \pdv{S_{k}}{\bs_{k,l}^{*}} $ follows. In particular, we have 
	\begin{align}
		\pdv{S_{k}}{\bs_{k,l}^{*}}=2\tr( \bar{\bR}_{k})\tr( \bar{\bR}_{k}'),\label{numerator1}
	\end{align}
	where $ (\cdot)' $ denotes the partial derivative with respect to $ \bs_{k,l}^{*}$. Based on the property  $ \tr\left(\bA \diag(\bs_{k,l}^{*})\right)=\left(\diag(A)\right)^{\T}\bs_{k,l}^{*} $, where $ \bA $ is a matrix independent of  $ \bs_{k,l}^{*}$ we can easily show 
	\begin{align}
		\tr( \bar{\bR}_{k}')=\beta_{k}\diag(\bG^{\H}\bG  \bPhi_{k}^{\X}\tilde{\bR}_{\mathrm{RIS}}^{\X} ).\label{tracederiv}
	\end{align}
	Regarding $ \pdv{I_{k}}{\bs_{k,l}^{*}} $, we obtain
	\begin{align}
		\!\!\!\pdv{I_{k}}{\bs_{k,l}^{*}}&\!=\!2\tr( \bar{\bR}_{k}'\bar{\bR}_{k})\!+\!{\sigma_{0}}{}\!\tr( \bar{\bR}_{k}')\!.\!\label{interference0}
	\end{align}
	The derivative in \eqref{int2}, is obtained by \eqref{interference0} after exploiting the property used in \eqref{tracederiv}.
	Substitution of the derivatives in \eqref{gam1} concludes the proof.
	\bibliographystyle{IEEEtran}
	
	\bibliography{IEEEabrv,mybib}

\begin{thebibliography}{10}
\providecommand{\url}[1]{#1}
\csname url@samestyle\endcsname
\providecommand{\newblock}{\relax}
\providecommand{\bibinfo}[2]{#2}
\providecommand{\BIBentrySTDinterwordspacing}{\spaceskip=0pt\relax}
\providecommand{\BIBentryALTinterwordstretchfactor}{4}
\providecommand{\BIBentryALTinterwordspacing}{\spaceskip=\fontdimen2\font plus
\BIBentryALTinterwordstretchfactor\fontdimen3\font minus
  \fontdimen4\font\relax}
\providecommand{\BIBforeignlanguage}[2]{{%
\expandafter\ifx\csname l@#1\endcsname\relax
\typeout{** WARNING: IEEEtran.bst: No hyphenation pattern has been}%
\typeout{** loaded for the language `#1'. Using the pattern for}%
\typeout{** the default language instead.}%
\else
\language=\csname l@#1\endcsname
\fi
#2}}
\providecommand{\BIBdecl}{\relax}
\BIBdecl

\bibitem{DiRenzo2020}
M.~Di~Renzo \emph{et~al.}, ``Smart radio environments empowered by
  reconfigurable intelligent surfaces: {How} it works, state of research, and
  the road ahead,'' vol.~38, no.~11, pp. 2450--2525.

\bibitem{Papazafeiropoulos2021}
A.~Papazafeiropoulos \emph{et~al.}, ``Intelligent reflecting surface-assisted
  {MU-MISO} systems with imperfect hardware: {Channel} estimation, beamforming
  design,'' \emph{IEEE Trans. Wireless Commun.}, 2021.

\bibitem{Huang2019}
C.~Huang \emph{et~al.}, ``Reconfigurable intelligent surfaces for energy
  efficiency in wireless communication,'' \emph{IEEE Transa. Wireless Commun.},
  vol.~18, no.~8, pp. 4157--4170, 2019.

\bibitem{Guo2020}
C.~{Guo} \emph{et~al.}, ``Outage probability analysis and minimization in
  intelligent reflecting surface-assisted {MISO} systems,'' \emph{IEEE Commun.
  Lett.}, vol.~24, no.~7, pp. 1563--1567, 2020.

\bibitem{Papazafeiropoulos2021a}
A.~Papazafeiropoulos \emph{et~al.}, ``Coverage probability of distributed {IRS}
  systems under spatially correlated channels,'' \emph{IEEE Wireless Commun.
  Lett.}, vol.~10, no.~8, pp. 1722--1726, 2021.

\bibitem{Zhi2021}
K.~Zhi \emph{et~al.}, ``Statistical {CSI}-based design for reconfigurable
  intelligent surface-aided massive {MIMO} systems with direct links,''
  vol.~10, no.~5, pp. 1128--1132.

\bibitem{Wu2019a}
Q.~Wu and R.~Zhang, ``Intelligent reflecting surface enhanced wireless network
  via joint active and passive beamforming,'' \emph{IEEE Trans. Wireless
  Commun.}, vol.~18, no.~11, pp. 5394--5409, 2019.

\bibitem{Bjoernson2020}
E.~{Bj{\"o}rnson} and L.~{Sanguinetti}, ``Rayleigh fading modeling and channel
  hardening for reconfigurable intelligent surfaces,'' \emph{IEEE Wireless
  Commun. Lett.}, vol.~10, no.~4, pp. 830--834, 2021.

\bibitem{Papazafeiropoulos2022}
A.~Papazafeiropoulos, ``Ergodic capacity of {IRS}-assisted {MIMO} systems with
  correlation and practical phase-shift modeling,'' \emph{IEEE Wireless Commun.
  Lett.}, vol.~11, no.~2, pp. 421--425, 2022.

\bibitem{Badiu2019}
M.-A. Badiu and J.~P. Coon, ``Communication through a large reflecting surface
  with phase errors,'' vol.~9, no.~2, pp. 184--188.

\bibitem{Xu2021}
J.~Xu \emph{et~al.}, ``{STAR-RISs}: {Simultaneous} transmitting and reflecting
  reconfigurable intelligent surfaces,'' \emph{IEEE Commun. Lett.}, vol.~25,
  no.~9, pp. 3134--3138, 2021.

\bibitem{Mu2021a}
X.~Mu \emph{et~al.}, ``Simultaneously transmitting and reflecting {(STAR) RIS}
  aided wireless communications,'' \emph{IEEE Trans. Wireless Commun.}, 2021.

\bibitem{Bohagen2007}
F.~Bohagen, P.~Orten, and G.~E. Oien, ``Design of optimal high-rank
  line-of-sight {MIMO} channels,'' \emph{IEEE Trans. Wireless Commun.}, vol.~6,
  no.~4, pp. 1420--1425, 2007.

\bibitem{Neumann2018}
D.~Neumann, M.~Joham, and W.~Utschick, ``Covariance matrix estimation in
  massive {MIMO},'' \emph{IEEE Signal Process. Lett.}, vol.~25, no.~6, pp.
  863--867, 2018.

\bibitem{Bjoernson2017}
E.~Bj{\"o}rnson \emph{et~al.}, ``Massive {MIMO} networks: Spectral, energy, and
  hardware efficiency,'' \emph{Foundations and Trends{\textregistered} in
  Signal Processing}, vol.~11, no. 3-4, pp. 154--655, 2017.

\bibitem{Boyd2004}
S.~Boyd, S.~P. Boyd, and L.~Vandenberghe, \emph{Convex optimization}.\hskip 1em
  plus 0.5em minus 0.4em\relax Cambridge university press, 2004.

\bibitem{3GPP2017}
3GPP, ``Further advancements for {E-UTRA} physical layer aspects ({Release}
  9),'' 3GPP TS 36.814, Tech. Rep., 2017.

\bibitem{Bai2015}
T.~Bai and R.~W. Heath, ``Coverage and rate analysis for millimeter-wave
  cellular networks,'' \emph{IEEE Trans. on Wireless Commun.}, vol.~14, no.~2,
  pp. 1100--1114, 2015.

\end{thebibliography}
\end{document}